\documentclass[a4paper,11pt]{article}
\usepackage[sort]{cite}
\usepackage{bookmark}
\usepackage[utf8]{inputenc} 
\usepackage{bussproofs}
\usepackage{enumerate}
\usepackage{enumitem} 
\usepackage[T1]{fontenc}  
\usepackage{amsmath,amssymb,amsthm}
\usepackage{mathtools}
\usepackage{enumitem}
\usepackage{hyperref}
\usepackage{stmaryrd}

\usepackage{lmodern} 
\usepackage{graphicx}
\usepackage{hyperref}
\usepackage{amssymb}
\usepackage{amsmath}
\usepackage{amsthm}
\usepackage{cite}
\usepackage{xcolor}
\usepackage[onehalfspacing]{setspace}
\usepackage{eurosym}

\usepackage{amsmath,amssymb,amsthm,units, stmaryrd,stackrel,relsize,bm}

\usepackage{bussproofs}
\usepackage{proof}
\usepackage[all]{xy}

\newtheorem{theorem}{Theorem}[section]

\makeatletter
\def\Ddots{\mathinner{\mkern1mu\raise\p@
\vbox{\kern7\p@\hbox{.}}\mkern2mu
\raise4\p@\hbox{.}\mkern2mu\raise7\p@\hbox{.}\mkern1mu}}
\makeatother

\newcommand{\e}{\varepsilon}
\newcommand{\limp}{\rightarrow}

\usepackage{bussproofs}
\usepackage{proof}

\newtheorem{definition}{Definition}

\newtheorem{example}{Example}
\newtheorem{remark}{Remark}

\usepackage{xcolor}

\usepackage{forest}
\usepackage{amssymb,amsmath}

\newcounter{tabline}

\forestset{
  tableau/.style={
    for tree={
      grow'=south,
      reversed=false, 
      parent anchor=south,
      child anchor=north,
      align=left,
      edge={-},
      l sep=10mm,
      s sep=8mm,
      math content
    }
  }
}

\title{Towards a Proof-Theoretic Analysis of Incorrect/Incomplete Proofs
\thanks{This research was funded in part by the Austrian Science Fund (FWF) 10.55776/P36571.}
}

\author{
Matthias Baaz, Mariami Gamsakhurdia\\
Technische Universität Wien\\
A–1040 Vienna, Austria\\
\textit{baaz@logic.at},
  \textit{mariami@logic.at}
 }

\begin{document}

\maketitle              

\begin{abstract}
We investigate the proof-theoretic structure of incorrect/incomplete proofs, that is, derivations containing syntactic errors or incomplete inferential steps that nonetheless preserve partial semantic validity. Building on Hilbert’s 
epsilon calculus, we formalize how such derivations can be corrected through semantic projection and weakest preconditions, leading to valid Herbrand disjunctions. We show that the 
epsilon calculus provides a natural framework for analyzing tolerance of falsity in proofs and for identifying conditions under which an incorrect proof can be semantically repaired. This approach extends Hilbert’s program beyond correctness, toward a logic of error and recovery. Moreover we show that the extended first epsilon theorem is false-tolerant.

\end{abstract}

\section{Introduction}
The notion of correctness has long been regarded as the boundary of what counts as a proof.
Yet, in practice, 
reasoning is rarely flawless, often precedes or exceeds formal correctness. 
In mathematics, proofs are often incorrect/incomplete, or based on false or insufficient premises: crucial steps are omitted, definitions are applied outside their strict scope, or intermediate claims are only heuristically justified \cite{Bookofproofs}, \cite{arneson2003witt}. 
Nevertheless, such proofs may capture a valid inferential structure that can later be repaired by filling in gaps or eliminating errors.

The aim of this paper is to develop a proof-theoretic framework capable of analyzing and correcting/completing such reasoning and to formalize this phenomenon in proof-theoretic terms.
We propose that Hilbert’s~$\varepsilon$-calculus, originally developed as a framework for eliminating quantifiers, provides an intrinsic mechanism for representing and repairing incorrect and/or incomplete reasoning. 
We investigate to what extent the syntactic elimination procedures of the ~$\varepsilon$-calculus retain information about the semantic defect of a non-valid ~$\varepsilon$-derivation.
In this sense, the~$\varepsilon$-calculus already anticipates a logic of semantic correction: a system that reconstructs from a possibly erroneous syntactic derivation a semantically valid statement.

The perspective developed here extends Hilbert’s program in a new direction.  
Instead of studying only the conditions under which a correct proof guarantees truth, we investigate the conditions under which an incorrect/incomplete proof can still yield truth once appropriately projected or corrected/completed.

We introduce the notion of \emph{weakest preconditions}, formalizing the minimal assumptions that render an incorrect proof valid and connecting them to propositional truth-table semantics.  
We then show that this mechanism built on $\varepsilon$-calculus captures a form of \emph{false tolerance}, absent in standard first-order logic.
The~$\varepsilon$-calculus thus becomes a calculus not only of existence but also of correction/completion, a system in which syntactic elimination can expose, preserve, or in some cases remove specific propositional defects of an erroneous derivation.


\section{Mathematical Motivation: Euler's Approach to the Basel Problem}\label{subsec:mathmotivation}

A well-known illustration of incomplete but profoundly insightful reasoning in mathematics is provided by Euler’s original argument for the Basel problem\cite{EulersSO}.
Euler became famous by deriving
\begin{equation}
 \sum^{\infty}_{n=1} \frac{1}{n^2} = \frac{\pi^2}{6}
\end{equation}
Let us follow Euler’s reasoning. Consider the polynomial of even degree
\begin{equation} \label{e2}
b_0 - b_1 x^2 + b_2 x^4 - \ldots + (-1)^n b_n x^{2n}
\end{equation}
If $b_n = 1$ it has the $2n$ roots $\pm \beta_1 , \ldots, \beta_n \not=0$ then (\ref{e2}) can be written as
\begin{equation}
(x - \beta_1 ) (x + \beta_1) \ldots (x - \beta_n ) (x + \beta_n )
\end{equation}
\begin{equation}
(-1)^n (\beta_1 -x) (\beta_1 + x) \ldots (\beta_n - x) (\beta_n +x )
\end{equation}
\begin{equation}
(-1)^n (\beta_{1}^2 - x^2 ) \ldots (\beta_{n}^2 - x^2)
\end{equation}
where $b_0 = (-1)^n \beta_{1}^2 \ldots \beta_{n}^2$
\begin{equation}\label{e3}
b_0 \left( 1 - \frac{x^2}{\beta_{1}^{2}} \right) \left( 1 - \frac{x^2}{\beta_{2}^2} \right) \ldots \left( 1 - \frac{x^2}{\beta_{n}^2} \right)
\end{equation}

    By comparing coefficients in (\ref{e2}) and (\ref{e3}) one obtains that
\begin{equation}\label{e4}
b_1 = b_0 \left( \frac{1}{\beta_{1}^{2}} + \frac{1}{\beta_{2}^{2}} + \ldots + \frac{1}{\beta_{n}^{2}} \right).
\end{equation}

Next Euler considers the Taylor series for $sin(x)$ divided by $x$
\begin{equation}\label{e5}
\frac{\mbox{sin }x}{x} = \sum_{n=0}^{\infty}(-1)^n \frac{x^{2n}}{(2n+1)\mbox{!}}
\end{equation}
which has as roots $\pm \pi, \pm 2\pi, \pm 3\pi, \ldots$.\\

Now by way of analogy Euler assumes\footnote{\textbf{Error}: Euler made an error in manipulating the infinite series without rigorously verifying its convergence (Euler assumed properties of finite polynomials hold for infinite series). He equated the Taylor expansion of \( \frac{\sin x}{x} \) with an infinite product. The correct value is $\zeta(2) = \frac{\pi^2}{6}$, proven later through Fourier analysis and more advanced techniques in analysis.} that the infinite degree polynomial (\ref{e5}) behaves in the same way as the finite polynomial (\ref{e2}). Hence in analogy to (\ref{e3}) he obtains
\begin{equation}\label{e6}
\frac{\mbox{sin } x}{x} = \left( 1- \frac{x^2}{\pi^2} \right) \left( 1- \frac{x^2}{4 \pi^2} \right) \left( 1- \frac{x^2}{9 \pi^2} \right) \ldots
\end{equation}
and in analogy to (\ref{e4}) he obtains
\begin{equation}\label{e7}
\frac{1}{3\mbox{!}} = \left( \frac{1}{\pi^2} + \frac{1}{4 \pi^2} + \frac{1}{9 \pi^2} + \ldots \right)
\end{equation}
which immediately gives
\begin{equation*}
 \sum^{\infty}_{n=1} \frac{1}{n^2} = \frac{\pi^2}{6}
\end{equation*}
Euler treated the function $\sin(x)$ as an infinite product and informally compared coefficients as if manipulating finite polynomials.  


From a modern standpoint, this argument is incorrect/incomplete: the passage from the power-series representation to an infinite product and the subsequent comparison of coefficients were not justified by the finite-polynomial argument alone.
Yet, when reconstructed in the framework of Weierstrass’s theory of entire functions, Euler’s reasoning can be seen as an \emph{anticipation} of valid results.
This historical episode illustrates the kind of \emph{semantic tolerance} we wish to capture proof-theoretically.  
Euler’s derivation contains syntactic errors (unjustified manipulations) but remains semantically meaningful: once projected into a more rigorous system, the argument yields a valid theorem.  
From a logical perspective, this corresponds to an \emph{incomplete proof} that can be transformed into a valid derivation.


\section{Weakest Preconditions}\label{sec:weakestpreconditions}
Given a formula $A$ occurring as the conclusion of an incorrect proof, we may ask:  
\emph{What minimal condition must hold for $B$ to be true, despite the proof’s defects?}  
The answer is the \emph{weakest precondition} of~$B$, denoted $\mathrm{wp}(B)$, representing the minimal assumption 
that restores the validity of the inference leading to~$B$.
The concept of \emph{weakest precondition} plays a central role in identifying minimal assumptions required to validate a proof. Intuitively,

\begin{definition}\label{def:weak}
    A precondition for a satisfiable formula  $B$ is any formula $A$ such that $A\rightarrow B$ is valid.
   We define a set $R$ of all satisfiable /admissible preconditions.
    We impose an order on $R$ among the satisfiable preconditions: $K < C$, meaning that $K\rightarrow C$ is valid but $C\rightarrow K$ is not valid.
A formula $A$ is a \emph{weakest precondition} for $B$ if:
\begin{itemize}[nosep]
    \item $A \rightarrow B$ is valid
    \item $A$ is satisfiable 
    \item $A$ is $<$- minimal among other satisfiable preconditions. 
    
    \end{itemize}
    
\end{definition}

Thus the weakest preconditions contain the reasonable minimal assumptions under which $B$ still holds. 
Such weakest preconditions are understood as preconditions implying other possible preconditions, and not implied by any of them, representing the simplest means of ensuring the correctness of a given argument.
\begin{remark}
We exclude $\bot$ (falsum) from being considered as a weakest precondition, as it trivially implies any formula, as $\bot \rightarrow B$ is always valid. Note that only 
$\bot$ has no precondition. The order $<$ might not be linear; there can be multiple incomparable weakest preconditions.

\end{remark}

This notion works well in propositional logic, over a fixed finite set of propositional variables, there are only finitely many formulae up to logical equivalence
, and hence minimal preconditions exist and can be computed via truth tables. In this setting, weakest preconditions correspond to conjunctions of literals that identify exactly one truth table line making the formula true.




\begin{example}
Given $S = (A \vee (B \rightarrow C)) \wedge (A \rightarrow B)$, we compute its truth table and identify the minimal (i.e., weakest) preconditions that make $S$ true. Each such conjunction of literals corresponding to a true line in the truth table is a weakest precondition.

   \begin{minipage}[t]{0.55\textwidth}
\centering
\begin{tabular}{|c|c|c|c|}
\hline
\textbf{A} & \textbf{B} & \textbf{C} & 
$(A\vee (B\rightarrow C))\wedge (A\rightarrow B)$
\\
\hline
0 & 0 & 0 & 1 \\ 
0 & 0 & 1 & 1 \\
0 & 1 & 0 & 0 \\
0 & 1 & 1 & 1 \\
1 & 0 & 0 & 0 \\
1 & 0 & 1 & 0\\
1 & 1 & 1 & 1 \\
1 & 1 & 0 & 1 \\
\hline 
\end{tabular}
\end{minipage}
\begin{minipage}[t]{0.3\textwidth}
   \textbf{Among the weakest preconditions are:} 
\begin{itemize}[nosep]
  \item $(\lnot A \land \lnot B \land \lnot C)$
  \item $(\lnot A \land \lnot B \land C)$
  \item $(\lnot A \land B \land C)$
  \item $(A \land B \land \lnot C)$
  \item $(A \land B \land C)$
\end{itemize}
\end{minipage}
\end{example}

In contrast to the propositional setting,
the situation in first-order logic is significantly more complex.
As an example, consider 
\[
\forall x\, (A(p(x)) \rightarrow A(x))\rightarrow A(0).
\]
Each formula $A(0)$, $A(p(0))$, $ A(p(p(0))) \dots$ is a valid precondition for a conclusion $A(0)$. However, there is no minimal such precondition, as these preconditions form an infinite descending chain under implication, i.e., each is strictly weaker than the previous one. This reflects a fundamental difference from propositional logic: the presence of quantifiers and function symbols in first-order logic leads to an unbounded space of logical dependencies. Thus, there is no accumulation point or limit of the chain of weakening preconditions.
The notion of \emph{weakest precondition} depends crucially on the underlying logical language and its expressiveness.

\begin{example}\label{derivation} Let us now consider incorrect proofs of an expression
   $(\forall x(A(x)\land \exists y B(y))\rightarrow \forall x A(x)) \land \exists x C(x)$.
   
    \begin{center}
\AxiomC{$A(a) \Rightarrow A(a)$}
\UnaryInfC{$A(a)\land \exists yB(y)\Rightarrow A(a)$}
\AxiomC{$\exists yB(y) \Rightarrow \exists yB(y)$} \RightLabel{$\land _l$}
\UnaryInfC{$\exists yB(y) \Rightarrow \exists yC(y)\quad *$} 
\RightLabel{$\land_l,\forall _l$}
\BinaryInfC{$\forall x(A(x)\land\exists yB(y))\Rightarrow A(a)$  \qquad $A(a)\land \exists yB(y)\Rightarrow \exists yC(y)$} \RightLabel{$\forall _l, \forall_r$}
\UnaryInfC{$\forall x(A(x)\land\exists yB(y))\Rightarrow \forall xA(x)$ \qquad $\forall x(A(x)\land \exists yB(y))\Rightarrow \exists yC(y)$} \RightLabel{$\land_r$} \LeftLabel{$\forall _l$} 
\UnaryInfC{$\forall x(A(x)\land\exists yB(y))\Rightarrow \forall xA(x) \land \exists yC(y)$}  \RightLabel{$\limp_r$}
\UnaryInfC{$\Rightarrow\forall x(A(x)\land\exists yB(y))\rightarrow \forall xA(x) \land \exists yC(y)$}
\DisplayProof
\end{center}
\end{example}
Where $*$ denotes incorrect inference step.
The obvious correction is obtained by $\exists yB(y)\rightarrow \exists yC(y)$, but also $\neg \forall x A(x)$  and $C(a)$ are preconditions. 
But how to calculate this in general? 
 But how to calculate a weakest precondition in general?
We need an alternative method: reduce first-order logic to propositional logic using Hilbert’s $\varepsilon$-calculus.

\section{Hilbert’s $\varepsilon$-calculus}\label{sec:epscalculus}
Hilbert's epsilon calculus \cite{hilbertbernays39},\cite{moser2006epsilon},\cite{MatthiasZach}  
is a first proof-theoretic formalism developed in the early 20th century as part of David Hilbert's program to provide a secure foundation for all of mathematics. At the heart of the epsilon calculus is a term-forming $\varepsilon$-operator intended to replace quantifiers with individual terms, thereby simplifying the manipulation of logical formulae. 
$\varepsilon _x\,A(x)$ that stands for “some $x$ satisfying property $A$ (if any).” 
This lets us reason with \emph{terms} instead of manipulating quantifiers step by step.

The \textit{standard translation} of first-order logic into epsilon calculus is defined recursively by mapping atomic formulae to themselves, respecting Boolean connectives, and dealing with quantified formulae as follows: suppose $A'$ is the standard translation of $A$, then we map:
\begin{align*}
\exists x\, A(x) \mapsto A'(\varepsilon_xA'(x)), &&
\forall x\, A(x) \mapsto  A'(\varepsilon_x\neg A'(x))   
\end{align*}
In epsilon calculus, critical formulae of the form \[
A(t) \rightarrow A(\varepsilon_ x A(x)) \quad A(\varepsilon_x\neg A(x))\mapsto  \, A(t)
\] where $t$ is an arbitrary term, play an important role as they express the dependency of an epsilon term on its possible witnesses and ensure correct behavior of replacement. Critical formulae are conservative over propositional logic.

\begin{example}
The standard translation of $\exists x A(x) \land \forall x B(x)$ is $A(\varepsilon_x A(x)) \land B(\varepsilon_x \neg B(x))$.
\end{example}

Note that in the translation to the $\varepsilon$-calculus the $\varepsilon$-terms should be abbreviated as $\varepsilon$-expression might become complicated. 
$[\exists x \exists y \exists z A(x,y,z)]^{\varepsilon } =$
\[ \begin{array}{l}
A(\varepsilon_x A(x,\varepsilon_y A(x,y,\varepsilon_z A(x,y,z)),\varepsilon_z A(x,\varepsilon_y
A(x,y,\varepsilon_z A(x,y,z)),z)), \\
\varepsilon_y A(\varepsilon_x A(x,\varepsilon_y A(x,y,\varepsilon_z A(x,y,z)),\varepsilon_z A(x,\varepsilon_y A(x,y,\varepsilon_z A(x,y,z)),z)),\\
y, \varepsilon_z A(\varepsilon_x A(x,\varepsilon_y A(x,y,\varepsilon_z A(x,y,z)),\varepsilon_z A(x,\varepsilon_y A(x,y,\varepsilon_z A(x,y,z))\\
,z)),y,z)),\varepsilon_z A(\varepsilon_x A(x,\varepsilon_y A(x,y,\varepsilon_z A(x,y,z)),\varepsilon_z A(x,\varepsilon_y A(x,y,\\
\varepsilon_z A(x,y,z)),z)),\varepsilon_y A(\varepsilon_x A(x,\varepsilon_y A(x,y,\varepsilon_z A(x,y,z)),\varepsilon_z A(x,\varepsilon_y A(x,\\
y,\varepsilon_z A(x,y,z)),z)),y,\varepsilon_z A(\varepsilon_x A(x,\varepsilon_y A(x,y,\varepsilon_z A(x,y,z)),\varepsilon_z A(x,\\
\varepsilon_y A(x,y,\varepsilon_z A(x,y,z)),z)),y,z)),z))
\end{array}\]

Now by abbreviation
\begin{align*}
        e_1(y,z)\equiv \varepsilon_x A(x,y,z)\\
       e_2(z)\equiv \varepsilon_y A(e_1(y,z),y,z)\\
       e_3\equiv \varepsilon_z A(e_1(e_2(z)),e_2(z),z)
    \end{align*}
    we obtain 
    $[\exists x \exists y \exists z A(x,y,z)]^{\varepsilon }\equiv A(e_1(e_2(e_3)),e_2(e_3),e_3). $

Regular first-order proofs are line-by-line translated into $\varepsilon$-calculus. For simplicity, we consider regular $LK$-proofs\cite{takeuti2013proof}:
The axiom sequents \(S\) become \([S]^\varepsilon\), structural and propositional rules are promoted unchanged,  weak quantifiers introduce critical formulae: \begin{itemize}[nosep]
  \item \(\exists\)-right: add \(A(t) \rightarrow A(\varepsilon_x A(x))\) to antecedent and 
    \item \(\forall\)-left: add \(\neg A(t) \rightarrow \neg A(\varepsilon_x \neg A(x))\) to antecedent
\end{itemize}
  Strong quantifiers are substituted by corresponding $\varepsilon$-terms \begin{itemize}[nosep]
    \item \(\exists\)-left: substitute \(a \mapsto \varepsilon_x A(x)\)
    \item \(\forall\)-right: substitute \(a \mapsto \varepsilon_x \neg A(x)\)
\end{itemize}
\begin{remark}
Valid $LK$-proofs (cut-free or otherwise), lead to valid $\varepsilon$-proofs\footnote{An $\varepsilon$-proof is valid if the conjunction of all critical formulae implies the conclusion, so that $\bigwedge (A_i(t_i)\rightarrow A_i(\varepsilon_xA_i(x)))\rightarrow A$ is a tautology in itself.}.
Note that  $\varepsilon$-proofs have no eigenvariable condition and contrary to usual first-order calculi, $\varepsilon$-calculus admits an unrestricted deduction theorem. 
Bounds for $\varepsilon$-proof depend only on the number of critical formulae and not on the propositional formulae.
\end{remark}

The following theorems play a crucial role for establishing soundness and completeness for  $\varepsilon$-calculus.
The following theorem describes the procedure that eliminates the critical formulae belonging to the same $\varepsilon$-term.

\begin{theorem} (Hilbert's Ansatz).
Given an $\varepsilon$-proof $\pi_\varepsilon$ of the form 
\begin{align*}
    [(A(t_1)\rightarrow A(\varepsilon_x A(x)))\wedge \dots \wedge( A(t_n)\rightarrow A(\varepsilon_x A(x)))]\rightarrow B(\varepsilon_x A(x))
\end{align*}
the resulting Herbrand's disjunction is 
    $$B(t_1)\vee \dots \vee B(t_n)\vee B(\varepsilon_x A(x))$$
where $\varepsilon_x A(x)$ can be replaced by 
one of the $t_i$ that does not contain an $\varepsilon$-term. In case all of the $t_i$ contain $\varepsilon$-term, one can replace only by some constant -obtaining a proof $\pi$ without $\varepsilon$-term.  
\end{theorem}

\begin{proof}
We use case distinction and the tautological nature of epsilon proofs: substitution of a term $t$ for $\varepsilon_x A(x)$ throughout a tautology preserves the tautological structure of $\varepsilon$-proof.
\end{proof}

The procedure can be extended to arbitrarily many types of critical formulae so that no critical formulae become non-critical (soundness) and the procedure terminates in the end.
For that reason we introduce the measures:
        \begin{enumerate}[nosep]
            \item Rank - the length of the overbinding chain (subordination),
            \item Degree - the depth of the nesting of $\varepsilon$-terms (inclusion),
        \end{enumerate}
where $\e$-term $e$ is nested in an $\e$-term $e'$ if $e$ is a proper subterm of $e'$ and $\e$-term $e$ is subordinate to $e'$ if $e$ occurs in $e'$.

\begin{theorem}(Extended First $\varepsilon$-theorem). An $\varepsilon$-proof of the translation of the existential formula can be stepwise transformed into a Herbrand's disjunction.
\end{theorem}

\begin{proof}
Hilbert's Ansatz is repeatedly applied to $\varepsilon$-terms of maximal rank (soundness) and maximal degree (termination) and to the remaining critical formulae. 
\end{proof}
 
The second $\varepsilon$-theorem allows one to reconstruct a proof of the original formulae from Herbrand's disjunction for the prenex formulae. Note that infix quantifiers can be reconstructed as suggested by Baaz, Hetzl, Weller \cite{Hetzl}.

\begin{theorem}(Second $\varepsilon$-theorem).
   If $A$ is a Herbrand disjunction of an existential formula after Skolemization\footnote{Skolemization is necessary to transform arbitrary formulae into existential formulae.}, we obtain a proof of the original prenex formula. 
\end{theorem}

Now let us go back to the Example \ref{derivation} of an incorrect proof and retranslate it into $\varepsilon$-calculus:\\
$e_1\equiv \varepsilon_x\neg A(x)$\\
$e_2\equiv \varepsilon_yB(y)$\\
$e_3\equiv \varepsilon_yC(y)$\\
 $e_4\equiv \varepsilon_x\neg (A(x)\land B(e_2))$
 \begin{center}
\AxiomC{ \qquad \qquad \qquad \qquad $B(e_2) \Rightarrow B(e_2)$}\RightLabel{$*$}
\UnaryInfC{ $A(e_1) \Rightarrow A(e_1)$ \qquad $B(e_2) \Rightarrow C(e_3)$} 
\UnaryInfC{$A(e_1)\land B(e_2) \Rightarrow A(e_1)$ \qquad $A(e_1)\land B(e_2) \Rightarrow C(e_3)$}
\UnaryInfC{$A(e_4)\land B(e_2) \Rightarrow A(e_1)$ \qquad $A(e_4)\land B(e_2) \Rightarrow C(e_3)$}
\UnaryInfC{$A(e_4)\land B(e_2) \Rightarrow A(e_1)\land C(e_3)$ }
\UnaryInfC{$\Rightarrow A(e_4)\land B(e_2) \rightarrow A(e_1)\land C(e_3)$ }
\DisplayProof 
 \end{center}
 The $\varepsilon$-proof is 
 \begin{center}
     $[\neg (A(e_1)\land B(e_2))\rightarrow \neg(A(e_4)\land B(e_2))]\limp A(e_4)\land B(e_2)\limp A(e_1)\land C(e_3).$
 \end{center}

\begin{minipage}[t]{0.4\textwidth}
\centering
\scriptsize
    \begin{tabular}{|c|c|c|c||c|c|c||c|}
\hline
$A_1$ & $B_2$ & $C_3$ & $A_4$ 
& $A_1 \land B_2 \rightarrow A_4 \land C_3$
& $A_4 \land B_2 \rightarrow A_1 \land C_3$ 
& Full Formula \\
\hline
0 & 0 & 0 & 0 & 1 & 1 & 1 \\
0 & 0 & 0 & 1 & 1 & 1 & 1 \\
0 & 0 & 1 & 0 & 1 & 1 & 1 \\
0 & 0 & 1 & 1 & 1 & 1 & 1 \\
0 & 1 & 0 & 0 & 1 & 1 & 1 \\
0 & 1 & 0 & 1 & 1 & 0 & 0 \\
0 & 1 & 1 & 0 & 1 & 1 & 1 \\
0 & 1 & 1 & 1 & 1 & 0 & 0 \\
1 & 0 & 0 & 0 & 1 & 1 & 1 \\
1 & 0 & 0 & 1 & 1 & 1 & 1 \\
1 & 0 & 1 & 0 & 1 & 1 & 1 \\
1 & 0 & 1 & 1 & 1 & 1 & 1 \\
1 & 1 & 0 & 0 & 0 & 0 & 1 \\
1 & 1 & 0 & 1 & 0 & 0 & 1 \\
1 & 1 & 1 & 0 & 0 & 0 & 1 \\
1 & 1 & 1 & 1 & 1 & 1 & 1 \\
\hline
\end{tabular}
\end{minipage}
\vspace{4mm}

A weakest precondition is e.g. 
 \begin{center}
    $A(e_1) \land A(e_4) \land  B(e_2)  \land C(e_3)$
 \end{center}
which retranslated to first-order language is 
\begin{center}
     $\forall xA(x) \land \forall x(A(x) \land \exists y B(y))\land \exists x C(x)$
 \end{center}

The following theorem highlights a remarkable property of the $\varepsilon$-calculus: even when the $\varepsilon$-proof is not a tautology, and thus the application of the $\varepsilon$-theorem does not yield a classically valid result, the degree of failure is preserved in a highly controlled way.



\begin{remark}
Various notions of weakest (or preferred) preconditions can be obtained by changing the underlying ordering (or preference) criterion; 
studying these alternatives is a topic in its own right. Howevevr, 
the term \emph{"weakest"} is used here with respect to logical implication:
a precondition \(P\) is weaker than \(Q\) if 
$P\rightarrow Q$ is valid but $Q\rightarrow P$ is not valid.

This semantic notion should not be confused with Boolean minimization,
as performed, for example, by the Quine--McCluskey algorithm \cite{Quine1952,McCluskey1956,Quine1955} optimizing Boolean formulas with respect to syntactic criteria.
Such
methods seek a smaller representation of the same Boolean function.

For instance,
\[
  (\neg x \land \neg y)\lor(\neg x\land y)
  \equiv \neg x,
\]
so replacing the formula on the left by \(\neg x\) is a minimization
of form, but not a genuine logical weakening. By contrast, replacing
the sufficient condition \(\neg x\land\neg y\) by the still-sufficient
condition \(\neg x\) is a genuine weakening, since
\[
  (\neg x\land\neg y )\rightarrow \neg x,
  \qquad
  \neg x \not\rightarrow (\neg x\land\neg y).
\]
Thus, our construction optimizes with respect to logical weakness,
rather than the syntactic size of a formula or circuit (minimality of representation).
 The central point of the present work, however, is that the relevant optimal preconditions exist and can be effectively calculated in the propositional setting. The role of the epsilon calculus is to extend this existence-and-construction result to first-order logic, where quantifiers and witness dependencies prevent a direct application of purely propositional optimization methods.
\end{remark}

\begin{remark}
  Note that sometimes the extended first epsilon theorem acts as a correcting device, more precisely: if a sentence does not contain $\varepsilon$-terms which correct the $\varepsilon$-proof, then this formula corrects Herbrand's disjunction. 
\end{remark}

  \begin{example}
        $(A(t)\supset A(e))\supset (A(e)\supset A(t))$ with $t$ not containing $e$, is not correct, but its Herbrand disjunction $A(t)\supset A(t)$ is correct.
    \end{example}

\begin{remark}

Note that number theories can be expressed using two types of critical formulae
\begin{itemize}[nosep]
    \item Standard critical formulae: $A(t) \rightarrow A(\varepsilon_ x A(x))$
    where $t$ contains no $\varepsilon$-terms.
    
    \item Inductive critical formulae: $A(t) \rightarrow A( \varepsilon _x A(x)\leq t )$
\end{itemize}
The second type introduces an arithmetic bound on the epsilon term and arises naturally in formulations of the induction schema. In fact, this idea played a historical role in Hilbert’s program: such bounded epsilon terms were used in early attempts to prove the consistency of number theory. 
In number theory, not all $\e$-proofs are predicative. Proofs involving bounded quantification (least-number principles) typically lead to bounded $\e$-terms. These may include $\e$-terms inside bounds, causing circularity and failure of $\e$-elimination.
The case distinction no longer terminates and the structure collapses or diverges based on whether $\e$-terms enter their own bounds. 
This shows that some arithmetic formulae do have Herbrand disjunctions, but only those corresponding to predicative $\e$-proofs - where every critical formula is of the form
\[
A(t) \rightarrow A(\varepsilon_x A(x))
\]
where the term $t$ contains no $\varepsilon$-terms at the moment of evaluation according to the Extended First $\varepsilon$-theorem.

\end{remark}

\section{Conclusion}\label{sec:conclusion}
The results of this paper are work in progress.
We computationally analyze incorrect and incomplete proofs in the epsilon calculus and investigate how weakest preconditions can be used to restore correctness. Our focus lies in connecting syntactic correction (via epsilon terms and their instances) with semantic interpretations (truth tables, conjunctions of literals), and in comparing with first-order and intuitionistic logic.

Note that errors in the epsilon proof are promoted to errors in the Herbrand disjunction in a controlled way. Under some circumstances the extended first epsilon theorem is extremely false-tolerant: degree of falsity is preserved.
We would like to conclude our paper with the most relevant open questions: what is the retranslation of an $\e$-expressions to first-order expressions?
   \begin{align*}
       P(\varepsilon_x\neg P(x,x), \varepsilon_x\neg P(x,x))\rightarrow P(\varepsilon_x\neg P(x,x), \underbrace{\varepsilon_y P(\varepsilon_x\neg P(x,x),y)}_{e})
   \end{align*} 
   $e$ has no meaning in first-order logic.
   Any satisfying answer implies the solution to the following question: Is there an elementary retranslation of $LK^\e$-proofs with cuts in $LK$-proofs with cuts (for a partial solution see \cite{Translations})?

%
%
%
%
\bibliographystyle{plain}
\bibliography{references}

\end{document}